\documentclass{llncs}

\usepackage{ifpdf}
\usepackage{latexsym}
\usepackage{graphicx}
\usepackage{color}
\usepackage{amsfonts}
\usepackage{amsmath}
\usepackage{amssymb}
\usepackage{slashbox}
\usepackage{multirow}
\usepackage{framed,xspace}
\usepackage{footmisc}

\ifpdf
  
\else
  
\fi

\newcommand{\cancel}[1]{}

\newcommand{\eps}{\varepsilon}

\DeclareMathOperator{\WCF}{WCF}
\DeclareMathOperator{\CF}{CF}

\newcounter{reffoot}

\begin{document}

\title{Tight Bounds for\\ Classical and Quantum Coin Flipping}

\author{Esther H{\"a}nggi\inst{1} \and J{\"u}rg Wullschleger\inst{2}}

\authorrunning{H{\"a}nggi, Wullschleger.}  

\tocauthor{Esther H{\"a}nggi, J{\"u}rg Wullschleger}

\institute{Computer Science Department, ETH Zurich, Z{\"u}rich, Switzerland\\
\and
DIRO, Universit\'e de Montr\'eal, Quebec, Canada \\
McGill University, Quebec, Canada}

\maketitle

\begin{abstract}
\emph{Coin flipping} is a cryptographic primitive for which strictly better protocols exist if the players are 
not only allowed to exchange classical, but also quantum messages.
During the past few years, several results have appeared which give a tight bound on the range of 
implementable unconditionally secure coin flips, both in the classical as well as in the quantum setting 
and for both weak as well as strong coin flipping. But the picture is still incomplete: in the quantum setting, 
all results consider only protocols with \emph{perfect correctness}, 
 and in the classical setting tight bounds for strong coin flipping are still missing. \\
 We give a general definition of coin flipping which unifies the notion of strong and weak coin 
flipping (it contains both of them as special cases) and allows the honest players to abort with 
a certain probability.  We give tight bounds on the achievable range of parameters both in the classical 
and in the quantum setting.
\end{abstract}

\section{Introduction}

\emph{Coin flipping} (or coin tossing)  as a cryptographic primitive has been introduced by Blum~\cite{blum} 
and is  one of the basic building blocks of \emph{secure two-party computation}~\cite{Yao82}. 

Coin flipping can be defined in several ways. The most common definition, sometimes called \emph{strong coin 
flipping}, allows two honest players to receive a uniform random bit $c \in \{0,1\}$, such that a dishonest player cannot \emph{increase} the probability of any output. A dishonest player may, however, abort the protocol, in which case the honest player gets the erasure symbol $\Delta$ as output\footnote{
The dishonest player may abort after receiving the output bit, but before the honest player gets the output bit. 
This allows cases where the honest player gets, for example, $0$ with probability $1/2$ and $\Delta$ otherwise.
There exists also a definition of coin flipping where a dishonest player does not have this unfair advantage, and the honest player must always get a uniformly random bit, no matter what the other player does. See~\cite{cleve,LoChau98,MoNaSe09}.
}. 
A weaker definition, called \emph{weak coin flipping}, only requires that each party cannot increase the 
probability of their preferred value.

Without any additional assumptions, unconditionally secure weak coin flipping (and therefore also strong coin 
flipping) cannot be implemented by a classical protocol. This follows from a result by Hofheinz, M\"uller-Quade 
and Unruh~\cite{HoMQUn06}, which implies that if two honest players always receive the same uniform bit, then 
there always exists one player that can force the bit to be his preferred value with certainty.

If the players can communicate using a quantum channel, unconditionally secure coin flipping is possible 
to some extent. The bounds of the possibilities have been investigated by a long line of research.
Aharanov \emph{et\ al.}~\cite{ATVY00} presented a strong coin flipping protocol where no quantum adversary 
can force the outcome to a certain value with probability larger than $0.914$. This bound has been improved 
by Ambainis~\cite{Ambain01} and independently by Spekkens and Rudolph~\cite{SpeRud01} to $0.75$ (see also~\cite{Colbeck} for a different protocol). 
For weak coin flipping, Spekkens and Rudolph~\cite{SpeRud02} presented a protocol where the dishonest player cannot force the outcome to its preferred value with probability larger than $1/\sqrt{2} \approx 0.707$. (Independently, Kerenidis and Nayak~\cite{KerNay04} showed a slightly weaker bound of $0.739$.)
This bound has further been improved by Mochon, first to $0.692$~\cite{Mochon04} and finally to $1/2 + \eps$ for any constant $\eps > 0$~\cite{Mochon07}, therefore getting arbitrarily close to the theoretical optimum. For strong coin flipping, on the other hand, this is not possible, since it has been shown by Kitaev~\cite{kitaev} (see~\cite{abdr} for a proof) that for any quantum protocol there is always a player able to force an outcome with probability at least $1/\sqrt{2}$. Chailloux and Kerenidis~\cite{ChaKer09} showed that a bound
of $1/\sqrt{2} + \eps$ for any constant $\eps > 0$ can be achieved{, by combining two classical protocols with Mochon's result: They first showed that an \emph{unbalanced} weak coin flip can be implemented using many instances of weak coin flips, and then that one instance of an unbalanced weak coin flip suffices to implement a strong coin flip with optimal achievable bias.}

\subsection{Limits of previous Results}

In all previous work on quantum coin flipping, honest players are required to output a \emph{perfect} coin flip, 
i.e., the probability of both values has to be exactly $1/2$, and the players must never disagree on the output 
or abort. However, in practice the players may very well be willing to allow a small probability of error even 
if both of them are honest. Furthermore, a (quantum) physical implementation of any protocol will always contain 
some noise and, therefore, also some probability to disagree or abort. This requirement is, therefore, overly 
strict and raises the question how much the cheating probability can be reduced when allowing an error of 
the honest players.

It is well-known that there exist numerous cryptographic tasks where allowing an (arbitrarily small) error can greatly improve the performance of the protocol.
 For example, as shown in~\cite{BeiMal04}, the amount of secure AND gates (or, alternatively, oblivious transfers) needed between two parties to test equality of two strings is only $O(\log 1/\eps)$ for any small error $\eps > 0$, while it is exponential in the length of the inputs in the perfect case.
Considering reductions from oblivious transfer to different variants of oblivious transfer where the players can use quantum communication, it has recently been shown 
that introducing a small error can reduce the amount of oblivious transfer needed by an arbitrarily large factor~\cite{WinWul10}.

It can easily be seen that \emph{some} improvement on the achievable parameters must be possible also in the case of coin flipping: In any protocol, the honest players can simply flip the output bit with some probability. This increases the error, but decreases the bias. In the extreme case, the two players simply flip two independent coins and output this value. This prohibits any bias from the adversary, at the cost of making the players disagree with probability $1/2$.

The only bounds on coin flipping we are aware of allowing for an error of the honest players have been given in 
the classical setting. An impossibility result for weak coin flipping has been given in~\cite{HoMQUn06}. Nguyen, 
Frison, Phan Huy and Massar presented in~\cite{singlecoin} a slightly more general bound and a protocol that 
achieves the bound in some cases.

\subsection{Contribution}

We introduce a general definition of coin flipping, characterized by $6$ parameters, which we denote by
\begin{align}
\nonumber &\CF(p_{00},p_{11},p_{0*},p_{1*},p_{*0},p_{*1})\;.
\end{align}
The value $p_{ii}$ (where $i \in \{0,1\}$)  is the probability that two honest players output $i$ and the value $p_{*i}$ ($p_{i*}$) is the maximal probability that the first (second) player can force the honest player to output $i$. With probability $1 - p_{00} - p_{11}$, two honest players will abort the protocol and output a dummy symbol.\footnote{Similar to~\cite{HoMQUn06}, we can require that two honest players always output the same values, since the players can always add a final round to check if they have the same value and output the dummy symbol if the values differ.}\setcounter{reffoot}{\value{footnote}}
This new definition has two main advantages:
\begin{itemize}
\item It generalizes both weak and strong coin flipping, but also allows for additional types of coin flips which are unbalanced or lay somewhere between weak and strong. 
\item It allows two honest players to abort with some probability.
\end{itemize}

We will first consider classical protocols (Section~\ref{sec:classical}), and give tight bounds for all parameters. The impossibility result (Lemma~\ref{lemma:classic-imposs}) uses a similar proof technique as Theorem~7 in~\cite{HoMQUn06}. In combination with two protocols showing that this bound can be reached  (Lemma~\ref{lemma:classicalProtocol}), we obtain the following theorem. 
\begin{theorem} \label{thm:class}
Let $p_{00},p_{11},p_{0*},p_{1*},p_{*0},p_{*1} \in [0,1]$. There exists a classical protocol that implements an unconditionally secure 
$\CF(p_{00},p_{11},p_{0*},p_{1*},p_{*0},p_{*1})$ if and only if
\begin{align*}
 p_{00} &\leq p_{0*} p_{*0}\;, \\
 p_{11} &\leq p_{1*}p_{*1}\;, \ \text{and} \\
 p_{00} + p_{11} &\leq p_{0*} p_{*0} + p_{1*}p_{*1} -  \max(0,p_{0*} + p_{1*} -1) \max(0,p_{*0} + p_{*1} -1) \;.
\end{align*}
\end{theorem}
For weak coin flipping, i.e.,  $p_{*1} = 1$ and $p_{0*}=1$, the bound of Theorem~\ref{thm:class} simplifies to the condition that $p_{00}\leq p_{*0}$, $p_{11}\leq p_{1*}$, and
\begin{align}
\nonumber  1 - p_{00} - p_{11} &\geq (1 - p_{*0}) (1- p_{1*} ) \;,
\end{align}
which is the bound that is also implied by Theorem~7 in~\cite{HoMQUn06}. 

In Section~\ref{sec:quantum}, we consider the quantum case, and give tight bounds for all parameters. The quantum protocol (Lemma~\ref{lemma:quantum-imp}) bases on one of the protocols presented in~\cite{ChaKer09}, and is a classical protocol that uses an unbalanced quantum weak coin flip as a resource. The impossibility result follows from the proof of Kitaev's bound on quantum strong coin flipping (Lemma~\ref{lemma:quant-imposs}). 
\begin{theorem}  \label{thm:quant}
Let $p_{00},p_{11},p_{0*},p_{1*},p_{*0},p_{*1} \in [0,1]$. For any $\eps > 0$, there exists a quantum protocol 
that implements an unconditionally secure 
$\CF(p_{00}, p_{11}, p_{0*} + \eps ,p_{1*} + \eps ,p_{*0}+\eps,p_{*1}+\eps)$ if 
\begin{align*}
 p_{00} &\leq p_{0*} p_{*0}\;, \\
 p_{11} &\leq p_{1*}p_{*1}\;, \ \text{and} \\
 p_{00} + p_{11} &\leq 1 \;.
\end{align*}
If these bounds are not satisfied then there does not exist a quantum protocol for $\eps = 0$.
\end{theorem}

\begin{figure}
\centering
{ \input{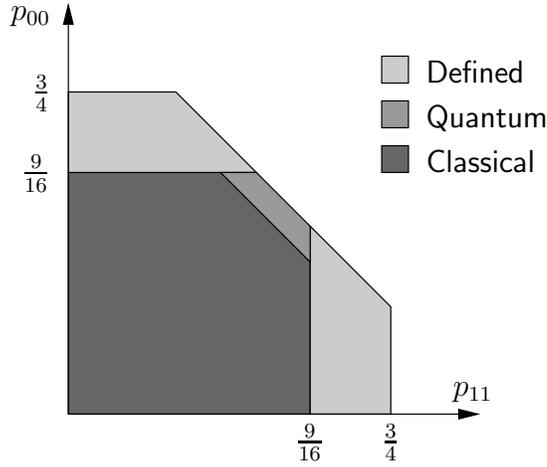} }
\caption{\label{fig:fig1}
For values $p_{0*} = p_{*0} = p_{1*} = p_{*1} = 3/4$, this figure shows the achievable values of $p_{00}$ and 
$p_{11}$ in the classical and the quantum setting. The light grey area is the set of all coin flips that can be 
defined. (See Definition \ref{def:coinflip}.)
}
\end{figure}

Our results, therefore, give the exact trade-off between weak vs.\ strong coin flipping, between bias vs.\ abort-probability, and between classical vs.\ quantum coin flipping. (Some of these trade-offs are shown in Figures~\ref{fig:fig1} and~\ref{fig:fig2}.)
They imply, in particular, that quantum protocols can achieve strictly better bounds if $p_{0*} + p_{1*} > 1$ and $p_{*0} + p_{*1} >1$. Outside this range classical protocols attain the same bounds as quantum protocols.

Since the optimal quantum protocol is a classical protocol using quantum weak coin flips as a resource, the possibility to do weak coin flipping, as shown by Mochon~\cite{Mochon07}, can be seen as the crucial difference between the classical and the quantum case.

\begin{figure}
\centering
 \input{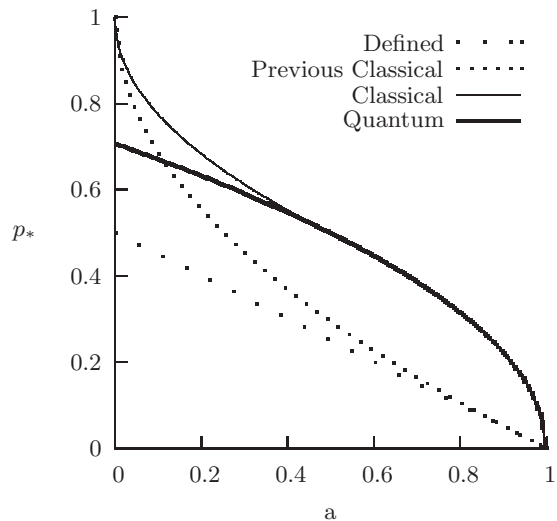}
\caption{\label{fig:fig2}
This graph shows the optimal bounds for symmetric coin flipping of the form $\CF((1 - a)/{2},(1 - a)/{2},p_{*},p_{*},p_{*},p_{*})$. The value $p_{*}$ is the maximal probability that any player can force the coin to be a certain value, and $a$ is the abort probability. Therefore, the smaller $p_{*}$ for a fixed value of $a$, the better is the protocol. The definition of coin flipping (Definition \ref{def:coinflip}) implies that $p_{*} \geq (1-a)/{2}$. Hence, the theoretically optimal bound is $p_{*} = (1-a)/{2}$.
In the quantum case, the optimal achievable bound is $p_{*} = \sqrt{(1-a)/2}$. In the classical case the optimal achievable bound is $p_{*} = 1 - \sqrt{a/2}$ for $a < 1/2$ and the same as the quantum bound for $a \geq 1/2$. The best previously known classical lower bounds from \cite{HoMQUn06,singlecoin} was $p_{*} \geq 1 - \sqrt{a}$.
}
\end{figure}

\section{Preliminaries}

In a classical protocol, the two players (Alice and Bob) are restricted to classical communication. Both players are given unlimited computing power and memory, and are able to locally sample random variables from any distribution. In a quantum protocol, the two players may exchange quantum messages. They have unlimited quantum memory and can perform any quantum computation on it. All operations are noiseless. At the beginning of the protocol, the players do not share any randomness or entanglement. While honest players have to follow the protocol, we do not make any assumption about the behaviour of the malicious players. We assume that the adversary is static, i.e., any malicious player is malicious from the beginning. Furthermore, we require that the protocol has a finite number of rounds.

\begin{definition}\label{def:coinflip}
Let $p_{00},p_{11},p_{0*},p_{1*},p_{*0},p_{*1} \in [0,1]$, such that $p_{00}+p_{11} \leq 1$, $p_{00} \leq \min \{p_{0*},p_{*0}\}$ and $p_{11} \leq \min \{p_{1*},p_{*1}\}$ holds.\footnote{The last two conditions are implied by the fact that a dishonest player can always behave honestly. Hence, he can always bias the coin to $i \in \{0,1\}$ with probability $p_{ii}$.} A protocol implements a $\CF(p_{00},p_{11},p_{0*},p_{1*},p_{*0},p_{*1})$, if the following conditions are satisfied:
\begin{itemize}
\item If both players are honest, then they output the same value $i \in \{0,1\}$ with probability $p_{ii}$ and $\Delta$ with probability $1 - p_{00} - p_{11}$.\footnotemark[\value{reffoot}]
\item For any dishonest Alice, the probability that Bob outputs $0$ is at most $p_{*0}$, and the probability that he outputs $1$ is at most $p_{*1}$.
\item For any dishonest Bob, the probability that Alice outputs $0$ is at most $p_{0*}$, and the probability that she outputs $1$ is at most $p_{1*}$.
\end{itemize}
\end{definition}

Definition~\ref{def:coinflip} generalizes the notion of both weak and strong coin flips and encompasses, in fact, the different definitions given in the literature. 
\begin{itemize}
\item A perfect weak coin flip is a
\begin{align}
\nonumber &\CF\left(\frac12,\frac12,1,\frac12,\frac12,1\right)\;.
\end{align}
\item A perfect strong coin flip is a
\begin{align}
\nonumber &\CF\left(\frac12,\frac12,\frac12,\frac12,\frac12,\frac12\right)\;.
\end{align}
\item
The weak coin flip with error $\eps > 0$ of~\cite{Mochon07} is a
\begin{align}
\nonumber &\CF\left(\frac12,\frac12,1,\frac12+\eps,\frac12+\eps,1\right)\;.
\end{align}
\item The unbalanced weak coin flip $\WCF(z, \eps)$ of~\cite{ChaKer09} is a
\begin{align}
\nonumber &\CF\left(z,1-z,1,1-z+\eps,z+\eps,1\right)\;.
\end{align}
\item The strong coin flip of~\cite{ChaKer09} is a 
\begin{align}
\nonumber &\CF\left(\frac12,\frac12,\frac {1} {\sqrt{2}} + \eps,\frac 1 {\sqrt{2}} + \eps,\frac 1 {\sqrt{2}} + \eps,\frac 1 {\sqrt{2}} + \eps\right)\;.
\end{align}
\end{itemize}

Note that $\CF(p_{00},p_{11},p_{0*},p_{1*},p_{*0},p_{*1})$ can also be defined as an ideal functionality that is equivalent to the above definition. Such a functionality would look like this: If there is any corrupted player, then the functionality first asks him to send a bit $b \in \{0,1\}$ that indicates which value he prefers. The functionality then flips a coin $c \in \{0,1,\Delta\}$, where the probabilities depend on $b$ and on which player is corrupted. For example, if the first player is corrupted and $b=0$, then $c=0$ will be chosen with probability $p_{*0}$, $c=1$ with probability $\min(p_{*1},1-p_{*0})$ and $\Delta$ otherwise. The functionality then sends $c$ to the adversary, and the adversary chooses whether he wants to abort the protocol or not. If he does not abort, the honest player receives $c$ (which might already be $\Delta$), and $\Delta$ otherwise.
If none of the players are corrupted, the functionality chooses a value $c \in \{0,1,\Delta\}$ which takes on $i \in \{0,1\}$ with probability $p_{ii}$ and sends $c$ to the two players.

\section{Classical Coin Flipping} \label{sec:classical}

\subsection{Protocols}

\begin{framed}
\noindent
Protocol \texttt{CoinFlip1}:\\
Parameters: $p_{0*},p_{1*},p_{*0},p_{*1} \in [0,1]$, $p_{0*} + p_{1*} \leq 1$.
\begin{enumerate}
 \item Alice flips a three-valued coin $a$ such that the probability that $a=i$ is $p_{i*}$ for $i=\{0,1\}$, and $a=\Delta$ otherwise. She sends $a$ to Bob. 
 \item If $a = \Delta$, Bob outputs $b=\Delta$. 
If $a \neq \Delta$, Bob flips a coin $b$ such that $b=a$ with probability $p_{*a}$
and $b=\Delta$ otherwise. Bob sends $b$ to Alice and outputs $b$.
\item If $b=a$ Alice outputs $b$, otherwise $\Delta$. 
\end{enumerate}
\end{framed}

\begin{lemma} \label{lem:protc1}
If either $p_{0*} + p_{1*} \leq 1$ or $p_{*0} + p_{*1} \leq 1$, then there exists a classical coin-flipping protocol with $p_{00} =  p_{0*} p_{*0}$ and $p_{11} = p_{1*}p_{*1}$.
\end{lemma}

\begin{proof}
If $p_{0*} + p_{1*} \leq 1$, they use Protocol \texttt{CoinFlip1}. (If $p_{*0} + p_{*1} \leq 1$, they exchange the role of Alice and Bob.)
By construction, a malicious Bob cannot bias Alice's output by more than $p_{i*}$, and a malicious Alice cannot bias Bob's output by more than
$p_{*i}$. 
Honest players output the value
$0$ with probability
$p_{0*} p_{*0}$ and $1$ with probability $p_{1*}p_{*1}$.
\qed
\end{proof}

\begin{framed}
\noindent
Protocol \texttt{CoinFlip2}:\\
Parameters: $p,x_0,x_1,y_0,y_1 \in [0,1]$.
\begin{enumerate}
 \item Alice flips a coin $a \in \{0,1\}$ such that $a=0$ with probability $p$ and sends it to Bob.
 \item Bob receives the coin $a$ and flips a coin $b \in \{0,1\}$ such that the probability that $b=a$ is $x_a$. He sends $b$ to Alice. If $b=a$ he outputs $b$. 
 \item If $b = a$, then Alice outputs $b$. If $a \neq b$, then Alice flips a coin $c$, such that with probability $y_b$, $c=b$ and else $c=\Delta$. She sends $c$ to Bob and outputs it. 
\item If $c=b$ Bob outputs $c$, else $\Delta$.
\end{enumerate}
\end{framed}

\begin{lemma} \label{lem:protc2}
If $p_{0*} + p_{1*} > 1$, $p_{*0} + p_{*1} > 1$, $p_{00} \leq p_{0*} p_{*0}$, $p_{11} \leq p_{1*} p_{*1}$, and 
\begin{eqnarray}
\label{eq:clcoin} p_{00} + p_{11} &=&  p_{0*} p_{*0} + p_{1*}p_{*1} - (p_{0*}+p_{1*}-1)(p_{*0} + p_{*1}-1)
\end{eqnarray}
then there exists a classical protocol implementing $\CF(p_{00},p_{11},p_{0*},p_{1*},p_{*0},p_{*1})$.
\end{lemma}

\begin{proof}
We use Protocol \texttt{CoinFlip2} and choose the parameters
\begin{align}
\nonumber x_i := p_{*i} \;, \quad 
y_0 := \frac{ p_{0*} - p}{1 - p}\;, \quad y_1 := \frac{p_{1*} + p - 1}{p}\;, \quad
p := \frac{p_{00} - p_{0*} + p_{0*}p_{*1}}{p_{*0} + p_{*1} - 1} \;.
\end{align}
Note that if $p=1$ then $y_0$ is undefined (and the same holds for 
$y_1$ if $p=0$), but this does not cause any problem since in this case the parameter 
$y_0$ is never used in the protocol. 

We now verify that these parameters are between $0$ and $1$. 
We have $y_0,y_1 \in [0,1]$, if $p \in [1 - p_{1*},p_{0*}]$. To see that $p$ lies indeed 
in this interval, note that the 
upper bound follows from
\begin{align}
\nonumber 
 p = \frac{p_{00} - p_{0*} + p_{0*}p_{*1}}{p_{*0} + p_{*1} - 1} \leq \frac{p_{0*}p_{*0} - p_{0*} + p_{0*}p_{*1}}{p_{*0} + p_{*1} - 1} = \frac{p_{0*} ( p_{*0} + p_{*1} - 1)}{p_{*0} + p_{*1} - 1} = p_{0*}\;.
\end{align}
For the lower bound, note that 
\begin{align}
\nonumber 1 - p
&= \frac{p_{*0} + p_{*1} - 1}{p_{*0} + p_{*1} - 1} - \frac{p_{00} -
p_{0*} + p_{0*} p_{*1}}{p_{*0} + p_{*1} - 1} \\
\nonumber&= \frac{p_{*0} + p_{*1} - 1 - p_{00} + p_{0*} - p_{0*} p_{*1}}{p_{*0}
+ p_{*1} - 1} \\
\label{eq:1minusp} &= \frac{p_{1*}p_{*0} -p_{1*}   + p_{11}}{p_{*0} + p_{*1} - 1}\;,
\end{align}
where we have used that
\begin{align*}
& p_{*0} + p_{*1} - 1 - p_{00} + p_{0*} - p_{0*} p_{*1} \\
 &\stackrel{(\ref{eq:clcoin})}{=} p_{*0} + p_{*1} - 1 - (p_{0*} p_{*0} + p_{1*}p_{*1} -
(p_{0*}+p_{1*}-1)(p_{*0} + p_{*1}-1) - p_{11})\\
\nonumber &\quad  + p_{0*} - p_{0*}
p_{*1} \\
&= p_{*0} + p_{*1} - 1 - p_{0*} p_{*0} - p_{1*}p_{*1} +
p_{0*}p_{*0} +p_{0*}p_{*1}- p_{0*} +p_{1*}p_{*0} +p_{1*}p_{*1}\\
 & \qquad
-p_{1*} -p_{*0} -p_{*1} +1   + p_{11}  + p_{0*} - p_{0*} p_{*1} \\
&= p_{1*}p_{*0} -p_{1*}   + p_{11}\;.
\end{align*}
Therefore
\begin{align*}
  p = 1 -  \frac{p_{11} - p_{1*} + p_{1*} p_{*0}}{p_{*0} + p_{*1} - 1} \geq
 1 -  \frac{p_{*1}p_{1*} - p_{1*} + p_{1*} p_{*0}}{p_{*0} + p_{*1} - 1} = 1 - p_{1*}\;.
\end{align*}
It follows that $p,x_0,x_1,y_0,y_1 \in [0,1]$.

If both players are honest, then the probability that they both output $0$ is
\begin{align*}
 p x_0 + (1-p) (1-x_1)y_0 
&= p x_0 + (1-p) (1-x_1) \frac{ p_{0*} - p}{1 - p}  \\
&= p p_{*0} +  (1-p_{*1}) ( p_{0*} - p)  \\
&= p p_{*0} - p (1-p_{*1}) +  p_{0*} (1-p_{*1}) \\
&= \frac{p_{00} - p_{0*} + p_{0*} p_{*1}}{p_{*0} + p_{*1} - 1} (p_{*0}+p_{*1}- 1) +  p_{0*} (1-p_{*1}) \\
&= p_{00}\;.
\end{align*}
The probability that they both output $1$ is
\begin{align*}
 p (1-x_0) y_1 + (1-p) x_1
& = p (1-p_{*0}) \frac{p_{1*} + p - 1}{p} + (1-p) p_{*1} \\
& = (1-p_{*0}) (p_{1*} + p - 1) + (1-p) p_{*1} \\
& =  p_{1*} (1-p_{*0}) - (1-p) (1-p_{*0}) + (1-p) p_{*1} \\
& =  p_{1*} (1-p_{*0}) + (1-p) (p_{*1} + p_{*0} - 1) \\
&\stackrel{(\ref{eq:1minusp})}{=}  p_{1*} (1-p_{*0}) + \frac{p_{1*}p_{*0} -p_{1*}   + p_{11}}{p_{*0} + p_{*1} - 1} 
 (p_{*1} + p_{*0} - 1) \\
& = p_{11}\;.
\end{align*}

If Alice is malicious, she can bias Bob to output value $i$ either by sending $i$ as first message hoping that Bob does not change the value, which has probability $x_i = p_{*i}$; or by sending the value $1-i$ hoping that Bob changes the value, which occurs with probability $1 - x_{1-i} = 1 - p_{*1-i} \leq p_{*i}$. Hence, she succeeds with probability $p_{*i}$.

Bob can bias Alice to output value $i$ by sending $b=i$ independently of what Alice had sent as first message. For $i=0$, Alice will accept this value with probability 
\begin{align*}
p + (1-p) y_0 = p + (1-p) \frac{ p_{0*} - p}{1 - p} = p_{0*}
\end{align*}
and for $i=1$ with probability
\begin{equation} 
1 - p + p y_1 = 1 -p + p \frac{p_{1*} + p - 1}{p} = p_{1*}\;.
\end{equation}
\qed
\end{proof}

In order to show that all values with $p_{00}+p_{11}$ below the bound given in (\ref{eq:clcoin}) can be reached, we will need additionally the following lemma. 
\begin{lemma} \label{lem:chelp}
If there exists a protocol~$P$ that implements a coin flip \linebreak[4] $\CF(p_{00},p_{11},p_{0*},p_{1*},p_{*0},p_{*1})$, then, for any $p'_{00} \leq p_{00}$ and $p'_{11} \leq p_{11}$, there exists a protocol $P'$ that implements a coin flip $\CF(p'_{00},p'_{11},p_{0*},p_{1*},p_{*0},p_{*1})$.
\end{lemma}

\begin{proof}
 $P'$ is defined as follows: The players execute protocol $P$. If the output is $i\in \{0,1\}$, then Alice changes to $\Delta$ with probability $1 - p'_{ii} / p_{ii}$. If Alice changes to $\Delta$, Bob also changes to $\Delta$. Obviously, the cheating probabilities are still bounded by $p_{0*},p_{1*},p_{*0},p_{*1}$, which implies that that protocol $P'$ implements a $\CF(p'_{00},p'_{11},p_{0*},p_{1*},p_{*0},p_{*1})$.
 \qed
\end{proof}

Combining
Lemmas~\ref{lem:protc1},~\ref{lem:protc2} and~\ref{lem:chelp}, we obtain Lemma~\ref{lemma:classicalProtocol}. 
\begin{lemma} \label{lemma:classicalProtocol}
Let $p_{00},p_{11},p_{0*},p_{1*},p_{*0},p_{*1} \in [0,1]$. There exists a classical protocol that implements
$\CF(p_{00},p_{11},p_{0*},p_{1*},p_{*0},p_{*1})$ if
\begin{align*}
 p_{00} &\leq p_{0*} p_{*0}\;, \\
 p_{11} &\leq p_{1*}p_{*1}\;, \ \text{and}  \\
 p_{00} + p_{11} &\leq p_{0*} p_{*0} + p_{1*}p_{*1} -  \max(0,p_{0*} + p_{1*} -1) \max(0,p_{*0} + p_{*1} -1) \;.
\end{align*}
\end{lemma}

\begin{proof}
 If $p_{0*} + p_{1*} >1$ and $p_{*0} + p_{*1} > 1$, then 
Lemmas~\ref{lem:protc2} and~\ref{lem:chelp} imply the bound. Otherwise, i.e., if either $p_{0*} + p_{1*} \leq 1$ or $p_{*0} + p_{*1} \leq 1$, then $\max(0,p_{0*} + p_{1*} -1) \max(0,p_{*0} + p_{*1} -1) = 0$ and the bound is implied by Lemmas~\ref{lem:protc1} and~\ref{lem:chelp}.
\qed
\end{proof}

\subsection{Impossibilities}

The following lemma shows that the bounds obtained in Lemma~\ref{lemma:classicalProtocol} are optimal. The proof uses the same idea as the proof of Theorem~7 in~\cite{HoMQUn06}.

\begin{lemma} \label{lemma:classic-imposs}
Let the parameters $p_{00},p_{11},p_{0*},p_{1*},p_{*0},p_{*1}$ be $\in [0,1]$.  A coin flip \linebreak[4] $\CF(p_{00},p_{11},p_{0*},p_{1*},p_{*0},p_{*1})$ can only be implemented by a classical protocol if
\begin{align*}
 p_{00} &\leq p_{0*} p_{*0}\;, \\
 p_{11} &\leq p_{1*}p_{*1}\;, \ \text{and}  \\
 p_{00} + p_{11} &\leq p_{0*} p_{*0} + p_{1*}p_{*1} - \max(0,p_{0*} + p_{1*} -1) \max(0,p_{*0} + p_{*1} -1)\;.
\end{align*}
\end{lemma}

\begin{proof}
We can assume that the output is a deterministic function of the transcript of the protocol. This can be enforced by adding an additional round at the end of the protocol where the two players tell each other what they are going to output. Since we do not require the protocol to be efficient, Lemma~7 in~\cite{HoMQUn06} implies that we can also assume that the honest parties maintain no internal state except for the list of previous messages.

 For any partial transcript $t$ of a protocol, we define $p^t_{0*}$ as the maximum over all transcripts starting with $t$, i.e., the maximum probability with which Bob can force Alice to output $0$, given the previous interaction has given $t$.
In the same way, we define $p^t_{1*}$, $p^t_{*0}$, $p^t_{*1}$. We define $p^t_{00}$ and $p^t_{11}$ as the probabilities that the output of the honest players will be $00$ and $11$, respectively, given the previous interaction has given $t$. We will now do an induction over all transcripts, showing that for all $t$, we have
\begin{align*}
 p^t_{00} &\leq p^t_{0*} p^t_{*0}\;, \\
 p^t_{11} &\leq p^t_{1*} p^t_{*1}\;, \ \text{and}  \\
 p^t_{00} + p^t_{11} &\leq p^t_{0*} p^t_{*0} + p^t_{1*}p^t_{*1} - \max(0,p^t_{0*} + p^t_{1*} -1) \max(0,p^t_{*0} + p^t_{*1} -1)\;.
\end{align*}

For complete transcripts $t$, each honest player will output either $0$, $1$ or $\Delta$ with probability $1$ and we always have
$p^t_{0*} + p^t_{1*} -1 \leq 0$ and $p^t_{*0} + p^t_{*1} -1 \leq 0$. Therefore, we only need to check that
$p^t_{00} \leq p^t_{0*} p^t_{*0}$ and $p^t_{11} \leq p^t_{1*} p^t_{*1}$. For $j \in \{0,1\}$, if $p^t_{jj} = 1$, then $p^t_{j*} = p^t_{*j} = 1$, so the condition is satisfied. In all the other cases we have $p^t_{jj} = 0$, in which case the condition is satisfied as well.

Let $t$ now be a partial transcript, and let Alice be the next to send a message. Let $M$ be the set of all possible transcripts after Alice has sent her message. For the induction step, we now assume that the statement holds for all transcript in $M$, and show that then it must also hold for $t$.
Let $r_i$ be the probability that an honest Alice will choose message $i \in M$. By definition, we have
\begin{align}
\nonumber  p^t_{00} = \sum_{i \in M} r_i p^i_{00}, \quad p^t_{11} = \sum_{i \in M} r_i p^i_{11},\quad 
p^t_{0*} = \sum_{i \in M} r_i p^i_{0*}, \quad p^t_{1*}
 = \sum_{i \in M} r_i p^i_{1*}\;,
\end{align}
\begin{align}
\nonumber p^t_{*0} = \max_{i \in M} p^i_{*0},  \quad p^t_{*1} = \max_{i \in M} p^i_{*1}
\;.
\end{align}
For $j \in \{0,1\}$ it holds that
\begin{align}
\nonumber  p^t_{jj} = \sum_{i \in M} r_i p^i_{jj} \leq \sum_{i \in M} r_i  p^i_{j*} p^i_{*j} \leq \sum_{i \in M} r_i  p^i_{j*} p^t_{*j} =  p^t_{j*} p^t_{*j}\;,
\end{align}
which shows the induction step for the first two inequalities. To show the last inequality, let 
\begin{align}
\nonumber  f(a,b,c,d) := ac + bd - \max(0,a+b -1) \max(0,c+d -1)\;,
\end{align}
where $a,b,c,d \in [0,1]$.
If we fix the values $c$ and $d$, we get the function
$f_{c,d}(a,b) := f(a,b,c,d)$. It consists of two linear functions: If $a+b \leq 1$, we have
\begin{align}
\nonumber  f_{c,d}(a,b) = ac + bd\;,
\end{align}
and if $a+b \geq 1$ we have
\begin{align}
\nonumber  f_{c,d}(a,b) = ac + bd - (a+b -1) \max(0,c+d -1)\;.
\end{align}
Note that these two linear functions are equal if $a+b=1$, and we have
$(a+b -1) \max(0,c+d -1) \geq 0$ if $a+b \geq 1$.
It follows that $f_{c,d}(a,b)$
 is concave, meaning that for all $\gamma,a,b,a',b' \in [0,1]$,
we have
\begin{align} \label{eq:concave}
  \gamma f_{c,d}(a,b) + (1-\gamma) f_{c,d}(a',b') \leq f_{c,d} (\gamma a+ (1-\gamma) a', \gamma b+ (1-\gamma) b'  )\;.
\end{align}

Since for any $a+b \neq 1$ and $c+d\neq 1$
\begin{align} \label{eq:monotone}
\frac{\partial }{\partial c} f(a,b,c,d) \geq  0 \qquad \mbox{and} \qquad \frac{\partial }{\partial d} f(a,b,c,d) \geq  0\;,
\end{align}
we have $f(a,b,c',d) \geq f(a,b,c,d)$ for $c' \geq c$ and
$f(a,b,c,d') \geq f(a,b,c,d)$ for $d' \geq d$. Hence,
\begin{align*}
 p^t_{00} &+ p^t_{11}\\
& = \sum_{i \in M} r_i (p^i_{00} + p^i_{11} ) \\
& \leq \sum_{i \in M} r_i  \left ( p^i_{0*} p^i_{*0} + p^i_{1*}p^i_{*1} -  \max(0,p^i_{0*} + p^i_{1*} -1)  \max(0,p^i_{*0} + p^i_{*1} -1) \right ) \\
& \leq \sum_{i \in M} r_i  \left ( p^i_{0*} p^t_{*0} + p^i_{1*}p^t_{*1} -  \max(0,p^i_{0*} + p^i_{1*} -1) \max(0,p^t_{*0} + p^t_{*1} -1) \right ) \\
& \stackrel{(\ref{eq:concave})}{\leq} p^t_{0*} p^t_{*0} + p^t_{1*}p^t_{*1} -  \max(0,p^t_{0*} + p^t_{1*} -1) \max(0,p^t_{*0} + p^t_{*1} -1)\;,
\end{align*}
and the inequalities also hold for $t$. The statement follows by induction.
\qed
\end{proof}

From Lemmas~\ref{lemma:classicalProtocol} and~\ref{lemma:classic-imposs} we obtain Theorem~\ref{thm:class}.

\section{Quantum Coin Flipping} \label{sec:quantum}

\subsection{Protocols}\label{subsec:qprotocols}

An \emph{unbalanced weak coin flip with error $\eps$} $\WCF(z,\eps)$ is a $\CF(z,1-z,1,1-z +\eps ,z +\eps,1)$, i.e., a coin flip where Alice wins with probability $z$, Bob with probability $1-z$ and both cannot increase their probability to win by more than $\eps$. (They may, however, decrease the probability to $0$.) Let $\WCF(z) := \WCF(z,0)$.

It has been shown by Mochon~\cite{Mochon07} that weak coin flipping can be implemented with an arbitrarily small error.

\begin{theorem}[Mochon~\cite{Mochon07}] \label{thm:wcf12}
For any constant $\eps > 0$, there exists a quantum protocol that implements $\WCF(1/2, \eps)$.
\end{theorem}

In~\cite{Mochon04a}, Mochon showed that quantum coin-flipping protocols compose sequentially. Implicitly using this result, Chailloux and Kerenidis showed that an unbalanced weak coin flip can be implemented from many instances of (balanced) weak coin flips.

\begin{proposition}[Chailloux, Kerenidis~\cite{ChaKer09}] \label{lem:wcf1}
For all $z \in [0,1]$, there exists a classical protocol that uses $k$ instances of $\WCF(1/2, \eps)$ and implements $\WCF(x, 2\eps)$, for a value $x \in [0,1]$ with $|x - z| \leq 2^{-k}$.
\end{proposition}

The following lemma shows that the parameter $z$ can be slightly changed without increasing the error much.

\begin{lemma} \label{lem:wcf2}
For any $1 > z' > z > 0$, there exists a classical protocol that uses $1$ instance of $\WCF(z', \eps)$ and implements $\WCF(z, \eps + z' - z)$.
\end{lemma}

\begin{proof}
The protocol first calls $\WCF(z', \eps)$. If Alice wins, i.e., if the output is $0$, then she changes the output bit to $1$ with probability $1 - z/{z'}$, and sends the bit to Bob. Bob only accepts changes from $0$ to $1$, but not from $1$ to $0$. Alice can force the coin to be $0$ with probability at most
$z' + \eps = z + (\eps + z' - z)$. Let $x \in [0,1 - z' + \eps]$ be the probability with which a cheating Bob
forces the protocol $\WCF(z', \eps)$ to output $1$. Alice will output $1$ with probability
\begin{align}
\nonumber  x + (1-x) \left (1 - \frac z {z'} \right )&= 1 - \frac z {z'} + x \cdot \frac z {z'}
\leq 1 - \frac z {z'} + (1 - z' + \eps) \cdot \frac z {z'}\\
\nonumber &= 1 - z + \eps \cdot \frac z {z'} \leq 1 - z + \eps\;.
\end{align}
\qed
\end{proof}

Note that for $z \in \{0,1\}$, the implementation of $\WCF(z, 0)$ is trivial. Hence,
Theorem~\ref{thm:wcf12}, Proposition~\ref{lem:wcf1} and Lemma~\ref{lem:wcf2} imply together that $\WCF(z, \eps)$ can be implemented for any $z \in [0,1]$ with an arbitrarily small error $\eps$.
To simplify the analysis of our protocols, we will assume that we have access to $\WCF(z)$ for any $z \in [0,1]$. The following lemma shows that when $\WCF(z)$ is replaced by $\WCF(z, \eps)$, the bias of the output is increased by at most $2\eps$.

\begin{lemma} \label{lem:shakeWCF}
Let $P$ be a protocol that implements $\CF(p_{00},p_{11},p_{0*},p_{1*},p_{*0},p_{*1})$ using  one instance of $\WCF(z)$. If $\WCF(z)$ is replaced by $\WCF(z,\eps)$, then $P$ implements 
 $\CF(p_{00},p_{11},p_{0*}+2\eps,p_{1*}+2\eps,p_{*0}+2\eps,p_{*1}+2\eps)$.
\end{lemma}

\begin{proof}
Let us compare two settings: one where the players execute $P$ using one instance of $\WCF(z,\eps)$, and the other where they use one instance of $\WCF(z)$. When both players are honest, the two settings are obviously identical. Let Alice be honest and Bob malicious. For each setting, we can define an event that occurs with probability at most $\eps$, such that under the condition that the two events do not occur, $\WCF(z)$ and $\WCF(z,\eps)$ and hence the whole protocol are identical. The probability that the two events do not occur is at least $1-2\eps$ by the union bound. Therefore, the probabilities that the honest player outputs $0$ (or $1$) differ by at most $2\eps$. The statement follows.
\qed
\end{proof}

The following protocol is a generalization of the strong coin-flipping protocol $S$ from~\cite{ChaKer09}. It gives optimal bounds for the case where the honest players never abort, i.e., $p_{00} + p_{11} = 1$. 
\begin{framed}
\noindent
Protocol \texttt{QCoinFlip1}:\\
Parameters: $x, z_0, z_1,p_0,p_1 \in [0,1]$.
\begin{itemize}
 \item Alice flips a coin $a \in \{0,1\}$ such that the probability that $a=0$ is $x$ and sends $a$ to Bob.
 \item Alice and Bob execute $\WCF(z_a)$.
 \item If Alice wins, i.e., the outcome is $0$, then both output $a$.
 \item If Bob wins, then he flips a coin $b$ such that $b=a$ with probability $p_a$. Both output $b$.
\end{itemize}
\end{framed}

\begin{lemma} \label{lem:q-noErr}
Let $p_{0*},p_{1*},p_{*0},p_{*1} \in [0,1]$ where $p_{*0}+p_{*1} > 1$, $p_{0*}+p_{1*} > 1$ and $ p_{*0} p_{0*} + p_{*1} p_{1*} = 1$.
Given access to one instance of $\WCF(z)$, we can implement a
$\CF(p_{00},p_{11},p_{0*},p_{1*},p_{*0},p_{*1})$
where $p_{00} = p_{0*} p_{*0}$ and  $p_{11} = p_{1*}p_{*1}$.
\end{lemma}

\begin{proof}
We execute Protocol \texttt{QCoinFlip1}, choosing the parameters
\begin{align*}
p_i := 1 - p_{*1-i}\;, \quad  z_0 &:= \frac{p_{*0}+p_{*1}-1}{p_{*1}}\;, \quad z_1 := \frac{p_{*0}+p_{*1}-1}{p_{*0}}\;,\\
 \mbox{and} \quad  x&:=\frac{p_{0*} p_{*0} + p_{*1}-1}{p_{*0}+p_{*1}-1}\;.
\end{align*}
Note that
\begin{align}
\nonumber  1 - z_0 = \frac{1 - p_{*0} }{ p_{*1} } \quad \mbox{and} \quad 1 - z_1 = \frac{1 - p_{*1}}{p_{*0}}\;. 
\end{align}
Since $1 - p_{*0} < p_{*1}$ and $1 - p_{*1} < p_{*0}$, these values are between $0$ and $1$, and hence also $z_0$ and $z_1$ are between $0$ and $1$.
From $p_{0*} \leq 1$ follows that $x \leq 1$, and from $p_{*0} p_{0*}+p_{*1} \geq p_{*0} p_{0*} + p_{*1}  p_{1*} = 1$ that $x \geq 0$.
Furthermore, we have
\begin{align} 
\label{eq:pstern0} z_0+(1-z_0) p_0 = \frac{p_{*0}+p_{*1}-1}{p_{*1}} + \frac{(1 - p_{*1})(1 - p_{*0})}{p_{*1}} = p_{*0}
\end{align}
and
\begin{align}
\nonumber  z_1+(1-z_1) p_1 = \frac{p_{*0}+p_{*1}-1}{p_{*0}} + \frac{(1 - p_{*1})(1 - p_{*0})}{p_{*0}} = p_{*1}\;.
\end{align}
Alice can bias Bob's coin to $0$ with probability
\begin{align*}
\max\lbrace{z_0+(1-z_0) p_0; (1-p_1) \rbrace} = p_{*0}
\end{align*}
and to $1$ with probability
\begin{align*}
\max\lbrace{z_1+(1-z_1) p_1; (1-p_0) \rbrace} = p_{*1}\;.
\end{align*}
The probability that Bob can bias Alice's coin to $0$ is
\begin{align*}
 x+(1-x) (1-z_1)
 & = (1-z_1) + x z_1 \\
 & = \frac{1 - p_{*1}}{p_{*0}} + \frac{p_{0*} p_{*0} + p_{*1}-1}{p_{*0}+p_{*1}-1} \cdot \frac{p_{*0}+p_{*1}-1}{p_{*0}} \\
 & = p_{0*}
\end{align*}
and the probability that he can bias it to $1$ is
\begin{align*}
(1-x)+x (1-z_0) 
 &= 1 - x z_0 \\
&\stackrel{(\ref{eq:pstern0})}{=}  1 - \frac{p_{0*} p_{*0} +p_{*1}-1}{p_{*0}+p_{*1}-1} \cdot \frac{p_{*0}+p_{*1}-1}{p_{*1}} \\
 &= 1 - \frac{p_{0*} p_{*0} +p_{*1}-1}{p_{*1}} \\
 &= \frac{1 -  p_{0*} p_{*0}}{p_{*1}} \\
 &= \frac{p_{1*} p_{*1}}{p_{*1}} = p_{1*}\;.
\end{align*}
Furthermore, two honest players output $0$ with probability
\begin{align*}
 x z_0+x (1-z_0) p_0+(1-x) (1&-z_1) (1-p_1)\\
&= x (z_0+ (1-z_0) p_0)+(1-x) \frac{1-p_{*1}}{p_{*0}} p_{*0}\\
&= x p_{*0} + (1-x) (1-p_{*1}) \\ 
&= 1-p_{*1} + x ( p_{*0} + p_{*1} - 1) \\
&= p_{0*}p_{*0} \\
&= p_{00}
\end{align*}
and $1$ with probability 
$1 - p_{00} = 1 - p_{0*}p_{*0} = p_{1*}p_{*1} = p_{11}$.
\qed
\end{proof}

The following protocol gives optimal bounds for the general case. It uses one instance of the above protocol, and lets Alice and Bob abort in some situations.

\begin{framed}
\noindent
Protocol \texttt{QCoinFlip2}:\\
Parameters: Protocol $P$, $\eps_0, \eps_1 \in [0,\frac12]$.
\begin{itemize}
 \item Alice and Bob execute the coin-flipping protocol $P$.
 \item If Alice obtains $0$, she 
changes to $\Delta$ 
with probability $\eps_0$. If Bob obtains $1$, he 
changes to $\Delta$ 
with probability $\eps_1$. If either Alice or Bob has changed to $\Delta$, they both output $\Delta$, otherwise they output the value obtained from $P$. 
\end{itemize}
\end{framed}

\begin{lemma} \label{lem:q-Err}
Let $p_{0*},p_{1*},p_{*0},p_{*1} \in [0,1]$ where $p_{*0}+p_{*1} > 1$, $p_{0*}+p_{1*} > 1$ and $p_{0*} p_{*0} + p_{*1} p_{*1} \leq 1$.
Given access to $\WCF(z)$ for any $z \in [0,1]$, we can implement a
$\CF(p_{00},p_{11},p_{0*},p_{1*},p_{*0},p_{*1})$
where $p_{00} = p_{0*} p_{*0}$ and  $p_{11} = p_{1*}p_{*1}$.
\end{lemma}

\begin{proof}
From $p_{*0}+p_{*1} > 1$ and $p_{0*}+p_{1*} > 1$ follows that either $p_{*0} + p_{1*} > 1$ or  $p_{0*} + p_{*1} > 1$. Without loss of generality, let us assume that $p_{*0} + p_{1*} > 1$.

Let
\begin{align}
\nonumber  p'_{0*} := \min \left ( 1, \frac{1 - p_{1*}p_{*1}}{p_{*0}} \right ) \quad \text{and} \quad p'_{*1} := \frac{ 1 - p'_{0*} p_{*0}} {p_{1*}} \;.
\end{align}
First, note that since $p_{0*} \leq \frac{1 - p_{1*}p_{*1}}{p_{*0}}$
we have $p'_{0*} \geq p_{0*}$. Obviously, we also have $p'_{0*} \leq 1$.
Since $p'_{0*} \leq \frac{1 - p_{1*}p_{*1}}{p_{*0}}$, we have
\begin{align} 
\nonumber p'_{*1} = \frac{ 1 - p'_{0*} p_{*0}} {p_{1*}} \geq \frac{ 1 - \frac{1 - p_{1*}p_{*1}}{p_{*0}} p_{*0}} {p_{1*}}
 = \frac{p_{1*}p_{*1}}{p_{1*}} = p_{*1}\;.
 \end{align}
 
In order to see that $p'_{*1} \leq 1$, we need to distinguish two cases. Since  $p'_{0*} := \min \left ( 1, \frac{1 - p_{1*}p_{*1}}{p_{*0}} \right )$, it holds that either $p'_{0*} = 1$ or $p'_{0*} =  \frac{1 - p_{1*}p_{*1}}{p_{*0}}$. In the first case, 
\begin{align} 
\nonumber  p'_{*1} &= \frac{ 1 - p_{*0}} {p_{1*}} < \frac{  p_{1*}} {p_{1*}} =1
\;,
\end{align}
and the claim holds. In the second case, 
\begin{align} 
\nonumber  p'_{*1} &= \frac{ 1 - p'_{*0}p_{*0}} {p_{1*}} =\frac{ 1 - (1 - p_{1*}p_{*1})} {p_{1*}}=p_{*1} \leq 1\;,
\end{align}
and the claim also holds. Therefore, $p'_{*1} \leq 1$.

Since $p'_{0*} p_{*0} +  p_{1*} p'_{*1} = 1$, according to Lemma~\ref{lem:q-noErr}, we can use protocol \texttt{QCoinFlip1}  to implement a
$\CF(p'_{00}, p'_{11}, p'_{0*},p_{1*},p_{*0},p'_{*1})$, where $p'_{00} = p'_{0*}p_{*0}$ and $p'_{11} = p_{1*}p'_{*1}$.
Using that protocol as protocol $P$, let Alice and Bob execute protocol \texttt{QCoinFlip2} with
 $\eps_0 := 1 - p_{0*} / p'_{0*}$, and $\eps_1 := 1 - p_{*1} / p'_{*1}$.

The probability that Bob can bias Alice to $0$ is now $(1-\eps_0) p'_{0*} = p_{0*}$, and the probability that Alice can bias Bob to $1$ is now $(1-\eps_1) p'_{*1} = p_{*1}$. Furthermore, the probability that two honest players output both $0$ is $(1-\eps_0) p'_{00} = (1-\eps_0) p'_{0*}p_{*0} = p_{0*}p_{*0}$  and the probability that they both output $1$ is $(1-\eps_1) p'_{11} = (1-\eps_1) p_{1*}p'_{*1} = p_{1*}p_{*1}$.
\qed
\end{proof}

\begin{lemma}\label{lemma:quantum-imp}
Let $p_{00},p_{11},p_{0*},p_{1*},p_{*0},p_{*1} \in [0,1]$ with
\begin{align*}
 p_{00} &\leq p_{0*} p_{*0}\;, \\
 p_{11} &\leq p_{1*}p_{*1}\;, \ \text{and}  \\
 p_{00} + p_{11} &\leq 1 \;.
\end{align*}
Then, for any constant $\eps > 0$, there exists a quantum protocol
that implements
$\CF(p_{00},p_{11},p_{0*}+ \eps,p_{1*}+ \eps,p_{*0}+ \eps,p_{*1}+ \eps)$.
\end{lemma}

\begin{proof}
Let us first assume that $p_{*0}+p_{*1} > 1$ and $p_{0*}+p_{1*} > 1$.
We reduce the value of $p_{0*}$ to $p_{00} / p_{*0}$ and the value of $p_{1*}$ to $p_{11} / p_{*1}$, which ensures that $p_{0*} p_{*0} + p_{1*} p_{*1} \leq 1$. Now we can apply Lemma~\ref{lem:q-Err}, together with Theorem~\ref{thm:wcf12}, Proposition~\ref{lem:wcf1} and Lemmas~\ref{lem:wcf2},~\ref{lem:shakeWCF} and~\ref{lem:chelp}.

If the assumption does not hold then either $p_{*0}+p_{*1} \leq 1$ or $p_{0*}+p_{1*} \leq 1$. In this case, we can apply Lemmas~\ref{lem:protc1} and~\ref{lem:chelp}.
\qed
\end{proof}

\subsection{Impossibilities}

In order to see that the bound obtained in Section~\ref{subsec:qprotocols} is tight, we can use the proof of Kitaev~\cite{kitaev} (printed in~\cite{abdr}) showing that an adversary can always bias the outcome of a strong quantum coin-flipping protocol. In fact, Equations (36) - (38) in~\cite{abdr} imply that for any quantum coin-flipping protocol, it must hold that 
$p_{11} \leq p_{1*}p_{*1}$. In the same way, it can be proven that $p_{00} \leq p_{0*} p_{*0}$. We obtain the following lemma. 
\begin{lemma} \label{lemma:quant-imposs}
 A $\CF(p_{00},p_{11},p_{0*},p_{1*},p_{*0},p_{*1})$ can only be implemented by a quantum protocol if
\begin{align*}
 p_{00} &\leq p_{0*} p_{*0}\;, \\
 p_{11} &\leq p_{1*}p_{*1}\;, \ \text{and}  \\
 p_{00} + p_{11} &\leq 1 \;.
\end{align*}
\end{lemma}
Lemma~\ref{lemma:quantum-imp} and~\ref{lemma:quant-imposs} imply together Theorem~\ref{thm:quant}. 

\section{Conclusions}

We have shown tight bounds for a general definition of coin flipping, which give trade-offs between weak vs.\ strong coin flipping, between bias vs.\ abort probability, and between classical vs.\ quantum protocols.

Our result extends the work of~\cite{ChaKer09}, and shows that the whole advantage of the quantum setting lies in the ability to do weak coin flips (as shown by Mochon~\cite{Mochon07}). If weak coin flips are available in the classical setting, classical protocols can achieve the same bounds as quantum protocols.

For future work, it would be interesting to see if similar bounds hold for the definition of coin flipping without the possibility for the malicious player to abort.

\subsubsection*{Acknowledgments.}
We thank Thomas Holenstein, Stephanie Wehner and Severin Winkler for helpful discussions, and the anonymous referees for their helpful comments. 
This work was funded by the Swiss National Science Foundation (SNSF), an ETHIIRA grant of ETH's research commission, the U.K. EPSRC, grant EP/E04297X/1 and the
Canada-France NSERC-ANR project FREQUENCY. Most of this work was done
while JW was at the University of Bristol.

\bibliographystyle{plain}
\bibliography{coin}

\end{document}